\documentclass[letterpaper,10pt]{article}
\usepackage[margin=1in]{geometry}

\usepackage[utf8]{inputenc} 
\usepackage[T1]{fontenc}    
\usepackage{url}            
\usepackage{booktabs}       
\usepackage{ragged2e}
\usepackage{amsfonts}       
\usepackage{nicefrac}       
\usepackage{microtype}      
\usepackage{comment}
\usepackage{amsthm,amsmath,amssymb,epsf, epsfig, color,enumerate,amssymb,ytableau}
\usepackage{algorithm}
\usepackage[noend]{algorithmic}
\usepackage{tikz}
\usepackage{subcaption}
\usepackage{caption} 
\usepackage{enumerate}
\usepackage{breqn}
\usepackage{algorithm}
\usepackage[noend]{algorithmic}
\usepackage[bookmarks=false]{hyperref}

\let\hide\iffalse
\let\unhide\fi

\begin{document}\sloppy

\let\hide\iffalse
\let\unhide\fi

\title{The Wedge Picking Model: A dynamic graph model based on triadic closure\footnote{This manuscript is a full presentation of what has appeared as ``Densest subgraph discovery in evolving graphs'' in the sets and partitions workshop, NeurIPS 2019 \& ``A theoretical analysis of graph evolution caused by triadic closure and algorithmic implications'' in IEEE BigData 2020.}}

\date{}
\author{Sara Ahmadian, Shahrzad Haddadan\footnote{Both authors contributed equally to this paper}}

\maketitle

\begin{abstract}
Social networks have become an inseparable part of human life and processing them in an efficient manner is a top priority in the study of networks. These networks are highly dynamic and they are growing incessantly. Inspired by the concept of \emph{triadic closure}, we propose a probabilistic mechanism to model the evolution of these dynamic graphs. Although triadic closure is ubiquitous in social networks and its presence helps forming communities, probabilistic models encapsulating it have not been studied adequately. 

We theoretically analyze our model and show how to bound the growth rate of some characteristics of the graph, such as degree of vertices. Leveraging our theoretical results, 
we develop a scheduling subroutine to process modifications of the graph in batches. Our scheduling subroutine is then used to speed up the state-of-the-art algorithms with negligible loss in their approximation guarantees. We demonstrate the applicability of our method by applying it to the densest subgraph and tri-densest subgraph discovery problem. 
\end{abstract}

\maketitle

\newcommand{\cohort}{cluster }
\newcommand{\cohorts}{clusters }
\newcommand{\tsD}{{\textsf{t-Star-Decomposition }}}



\newcommand{\EX}{{\cal X}}
\newcommand{\J}{{\cal J}}
\newcommand{\N}{{\mathbb{N}}}
\newcommand{\R}{\mathbb{R}}
\newcommand{\Ho}{{\cal H}}
\newcommand{\M}{{\cal M}}
\newcommand{\Lc}{{\cal L}}




\newtheorem{defin}{Definition}[section]
\newtheorem{lemma}{Lemma}[section]
\newtheorem{thm}{Theorem}[section]
\newtheorem{prop}[thm]{Proposition}
\newtheorem{claim}[thm]{Claim}
\newtheorem{remark}[thm]{Remark}
\newtheorem{con}[thm]{Corollary}
\newtheorem{conj}[thm]{Conjecture}
\newtheorem{exm}[thm]{Example}
\newtheorem{prob}[thm]{Problem}
\newcommand{\arx}[1]{\href{http://arxiv.org/abs/#1}{\texttt{arXiv:#1}}}

\renewcommand\qedsymbol{$\blacksquare$}
\renewcommand{\algorithmicrequire}{\textbf{Input:}}
\renewcommand{\algorithmicensure}{\textbf{Output:}}



\newcommand{\eps}{{\epsilon}}
\newcommand{\al}{\ensuremath{\alpha}}
\newcommand{\dt}{\ensuremath{\delta}}
\newcommand{\sg}{\ensuremath{\sigma}}
\newcommand{\ld}{\ensuremath{\lambda}}
\newcommand{\Om}{\ensuremath{\Omega}}
\newcommand{\e}{\ensuremath{\epsilon}}
\newcommand{\ve}{\ensuremath{\varepsilon}}
\newcommand{\gm}{\ensuremath{\gamma}}
\newcommand{\Gm}{\ensuremath{\Gamma}}
\newcommand{\kp}{\ensuremath{\kappa}}
\newcommand{\Tau}{\mathrm{T}}

\newcommand{\cA}{\ensuremath{\mathcal{A}}}
\newcommand{\cCB}{\ensuremath{\mathcal{CB}}}
\newcommand{\cC}{\ensuremath{\mathcal{C}}}
\newcommand{\cD}{\ensuremath{\mathcal D}}
\newcommand{\cF}{\ensuremath{\mathcal F}}
\newcommand{\cG}{\ensuremath{\mathcal{G}}}
\newcommand{\cI}{\ensuremath{\mathcal I}}
\newcommand{\cL}{\ensuremath{\mathcal L}}
\newcommand{\cN}{\ensuremath{\mathcal N}}
\newcommand{\cP}{\ensuremath{\mathcal P}}
\newcommand{\cS}{\ensuremath{\mathcal S}}
\newcommand{\cU}{\ensuremath{\mathcal U}}

\newcommand{\bR}{\ensuremath{\mathbb R}}
\newcommand{\bZ}{\ensuremath{\mathbb Z}}
\newcommand{\bQ}{\ensuremath{\mathbb Q}}
\newcommand{\bN}{\ensuremath{\mathbb N}}

\newcommand{\RR}{\mathbb{R}}
\newcommand{\RRN}{\mathbb{R}_{\ge 0}}
\newcommand{\QQ}{\mathbb{Q}}
\newcommand{\ZZ}{\mathbb{Z}}
\newcommand{\ZZN}{\mathbb{Z}_{\ge 0}}

\newcommand{\0}{\ensuremath{\pmb{0}}}
\newcommand{\1}{\ensuremath{\pmb{1}}}
\newcommand{\g}{\ensuremath{\pmb{g}}}
\newcommand{\x}{\ensuremath{\pmb{x}}}
\newcommand{\y}{\ensuremath{\pmb{y}}}

\newcommand*{\pr}{^{\prime}\mkern-1.2mu}
\newcommand*{\dpr}{^{\prime\prime}\mkern-1.2mu}
\newcommand*{\tpr}{^{\prime\prime\prime}\mkern-1.2mu}

\newcommand{\ceil}[1]{\ensuremath{\left\lceil#1\right\rceil}}
\newcommand{\floor}[1]{\ensuremath{\left\lfloor#1\right\rfloor}}
\newcommand{\poly}{\operatorname{\mathsf{poly}}}
\newcommand{\polylog}{\operatorname{polylog}}

\newcommand{\sm}{\ensuremath{\setminus}}
\newcommand{\es}{\ensuremath{\emptyset}}
\newcommand{\sse}{\subseteq}
\newcommand{\frall}{\ensuremath{\text{ for all }}}
\newcommand{\freach}{\ensuremath{\text{ for each }}}


\newcommand{\OPT}{\mbox{\sc OPT}}
\newcommand{\lp}{{\small \textsf{LP}}\xspace}
\newcommand{\lps}{{\small \textsf{LP}s}\xspace}
\newcommand{\opt}{\ensuremath{{\mathit{opt}}}}
\newcommand{\iopt}{\ensuremath{O^*}}
\newcommand{\lpopt}{\ensuremath{\mathit{OPT}}}


\newcommand{\p}{{\em P}\xspace}
\newcommand{\np}{{\em NP}\xspace}
\newcommand{\nphard}{\np-hard\xspace} 
\newcommand{\npcomplete}{\np-complete\xspace}
\newcommand{\apx}{{\em APX}\xspace}
\newcommand{\apxhard}{\apx-hard\xspace}

\newcommand{\minsat}{\ensuremath{\mathsf{MIN}\ \mathsf{SAT}}\xspace}


\newcommand{\expct}[1]{\ensuremath{\text{{\bf E}$\left[#1\right]$}}}
\newcommand{\size}[1]{\ensuremath{\left|#1\right|}}


\newcommand{\dij}{\ensuremath{d(i,j)}}
\newcommand{\dijp}{\ensuremath{d(i,j')}}
\newcommand{\dipj}{\ensuremath{d(i',j)}}
\newcommand{\dipjp}{\ensuremath{d(i',j')}}

\newcommand{\dir}{\ensuremath{d(i,r)}}
\newcommand{\di}{\ensuremath{d_i}}

\newcommand{\z}{\ensuremath{z}}

\newcommand{\f}{\ensuremath{f}}

\newcommand{\cost}{\ensuremath{\mathit{cost}}}
\newcommand{\maxcost}{\ensuremath{\mathit{maxr}}}
\newcommand{\charge}{\ensuremath{\mathit{charge}}}
\newcommand{\weight}{\ensuremath{\mathit{weight}}}
\newcommand{\val}{\ensuremath{\mathit{value}}}




\newcommand{\ufl}{\ensuremath{\mathsf{UFL}}\xspace}
\newcommand{\fl}{\ensuremath{\mathsf{FL}}\xspace}
\newcommand{\cfl}{{\small \textsf{CFL}}\xspace}
\newcommand{\mfl}{{\small \textsf{MFL}}\xspace}
\newcommand{\load}{\ensuremath{\mathsf{Load}}}
\newcommand{\ml}{\ensuremath{\mathsf{Max\text{-}Load}}}
\newcommand{\mlkfl}{\ensuremath{\mathsf{ML}k\mathsf{FL}}\xspace}
\newcommand{\mmfl}{\ensuremath{\mathsf{MFL}}}
\newcommand{\lbkc}{\ensuremath{\mathsf{LB}k\mathsf{Cent}}\xspace}
\newcommand{\lbks}{\ensuremath{\mathsf{LB}k\mathsf{Sup}}\xspace}
\newcommand{\ulbkc}{\ensuremath{\mathsf{ULB}k\mathsf{C}}\xspace}
\newcommand{\mlsc}{\mlkfl\xspace}
\newcommand{\lbfl}{\ensuremath{\mathsf{LBFL}}\xspace}
\newcommand{\ulbfl}{{\small \textsf{ULBFL}}\xspace}

\newcommand{\lbkco}{\ensuremath{\mathsf{LB}k\mathsf{CentO}}\xspace}
\newcommand{\lbkso}{\ensuremath{\mathsf{LB}k\mathsf{SupO}}\xspace}
\newcommand{\ckco}{\ensuremath{\mathsf{C}k\mathsf{CO}}\xspace}
\newcommand{\ckso}{\ensuremath{\mathsf{C}k\mathsf{SO}}\xspace}
\newcommand{\lbksr}{\ensuremath{\mathsf{LB}k\mathsf{SR}}\xspace}
\newcommand{\lbksro}{\ensuremath{\mathsf{LB}k\mathsf{SRO}}\xspace}
\newcommand{\kbs}{\ensuremath{k\text{-}\mathsf{BS}}\xspace}
\newcommand{\ksr}{\kbs}
\newcommand{\flsr}[1]{\ensuremath{\mathsf{FLSR}(#1)}\xspace}
\newcommand{\cc}{\ensuremath{\mathsf{CellC}}\xspace}

\newcommand{\rad}{\ensuremath{\mathsf{rad}}\xspace}
\newcommand{\Rad}{\ensuremath{\mathsf{Rad}}\xspace}
\newcommand{\uncov}{\ensuremath{\mathsf{uncov}}\xspace}
\newcommand{\uc}{\ensuremath{\mathsf{uc}}\xspace}
\newcommand{\outl}{\ensuremath{\mathit{Out}}\xspace}
\newcommand{\AT}{\ensuremath{\mathit{AT}}\xspace}
\newcommand{\AD}{\ensuremath{\mathit{AD}}\xspace}
\newcommand{\dist}{\ensuremath{\mathit{dist}}\xspace}

\newcommand{\pd}{\ensuremath{PD}\xspace}
\newcommand{\pdalg}{\ensuremath{\mathsf{PDAlg}}\xspace}
\newcommand{\pdo}{\ensuremath{\mathsf{PDAlg^o}}\xspace}
\newcommand{\kbalg}{\ensuremath{k\text{-}\mathsf{BSAlg}}\xspace}
\newcommand{\kbo}{\ensuremath{k\text{-}\mathsf{BSAlg^o}}\xspace}

\newcommand{\OPTsr}{\ensuremath{\sc OPT_{\sum r}}}
\newcommand{\optsr}{\ensuremath{{\mathit{opt_{\sum r}}}}}

\section{Introduction}

Structural properties of social network distinguish them from other graphs. These properties include  bounded degeneracy, heavy tailed degree distribution, high density of triangles. 
Various  generative models have been suggested in the literature to study or explain these properties. The most celebrated of these generative models are the preferential attachment model \cite{PrefModel}, the copying model  \cite{Kumar:2000:SMW:795666.796570}, or  Chung-Lu random graph model \cite{chungmodel}.  Each of these models explain some property of social networks but not all of them; see  \cite{Chakrabarti:2006} for a comparative survey. Among all of these models, only a few have encapsulated an important property of social networks: \emph{ triadic closure} \cite{Robins2009ClosureCA,Brunson2015TriadicAO,Bianconi:1741450}.  
\smallskip 

Put simply, triadic closure is the tendency of a disconnected pair of users with a large common neighborhood to form connection. 
German sociologist George Simmel coined the term in 1908, suggesting that if there is a strong connection between A and B and there is a strong connection between A and C, then there must exist a strong or weak connection between B and C.

\smallskip 

Triadic closure is a prominent ingredient  for formation of modularity in networks \cite{mod} and it has been shown that most of generative models such as  Barab\'asi-Albert model do not generate modularity without additionally  enforcing triadic closure \cite{Bianconi:1741450}.  Thus, deep mathematical understanding of models explaining triadic closure is crucial.  
Nonetheless, the generative models based on triadic closure have not been theoretically analyzed, and little information about their growth rate or other properties is known.

\smallskip

Among theoretical computer scientists, a new trend of studying social networks known as ``distribution free models'' has recently emerged  \cite{roughgarden2020distributionfree}. Such studies leverage on specific properties of social networks to design tractable algorithms for expensive and in particular NP-hard problems. In this context the 
``c-closure''  property was introduced by Fox et al \cite{Roughgarden2018}, and it sparked notable research \cite{cclose1, cclose2, cclose3}. A c-closed graph is a graph in which any two vertices $v$ and $u$ having a common neighborhood larger than $c$ are connected. Problems such as listing all maximal cliques, or computing dense or sparse subgraphs, which are NP-hard in general are tractable for c-closed graphs \cite{Roughgarden2018, koana2020}. Distribution free models  and in particular c-closure have an important limitation: they fall short in modeling dynamic graphs which are more realistic in the context of social networks. 

\smallskip

Inspired by Fox et al and subsequent work on c-closed graphs, our goal in this paper is to introduce a probabilistic  dynamic model based on triadic closure, and analyze it theoretically with the goal of designing algorithms.

In particular,   we  seek answer to the following questions:

\begin{enumerate}
    \item 
 How does triadic closure explain the evolution of dynamic graphs? 

\item Can we theoretically analyze the effect of triadic closure on the growth rate of degree of vertices of a dynamic graph?

\item Can we use such bounds to  design better algorithms for dynamic graphs (or to maintain previously calculated solutions)?
\end{enumerate}

We provide answers to the above question by (1) introducing the \emph{wedge picking model},  a probabilistic model for dynamic graphs based on  triadic closure (2) presenting a theoretical analysis of the wedge picking model and in particular the growth rate of its local properties such as the degree of each vertex (3) after verifying that the wedge picking model explains the evolution of an input dynamic graph and estimating the parameters of the model in a pre-processing phase, we use the above theoretical results to design a speed up subroutine which we call the \emph{rest and run} strategy. We  employ our proposed method to an important primitive of graph mining: densest subgraph discovery. 
 
  \smallskip 
  
 The rest and run strategy works as an scheduling subroutine determining periods of ``rest'' and ``run'' for some optimization problems, thus improving their run time. The idea is to identify pre-calculated intervals for ``resting'' while theoretically ensuring that with high probability the quality of solution will not degrade by much. After the ``rest'' phase, we read a batch of updates and start the ``run'' phase in which we run the original algorithm.
 
\subsection*{Related Work} 
Study of probabilistic models for networks goes back to 1959 when pioneers of graph algorithms, Erd\"os and R\'enyi, introduced their model. Since Erd\"os–R\'enyi model does not explain most of the properties of real networks, other models like Watts–Strogatz model \cite{WWmodel} have been introduced. The Watts–Strogatz model captures the ``six degree of seperation'' in real network: despite having low average degree, most of  vertices are reachable from each other via at most six hops. Their model also captures  the existence of communities.

\smallskip 

Real world networks have been analyzed vastly with the motivation to identify properties  distinguishing them from regular random networks. Such analysis can help design of generative models.  For instance, Ugander et al. \cite{Ugander2011} analyzed  social graph of active Facebook users, and reported properties such as: degree distribution, path length, component size, clustering coefficient, degeneracy, degree correlation, etc.  
Various models have been proposed for explaining different properties of such networks (see, e.g., \cite{Chakrabarti:2006}). Only to explain degree distribution various models have been suggested which are all based on the ``rich gets richer'' phenomenon: Barab\'asi-Albert \cite{Barbarasi} suggests the preferential attachment model which generates graphs with power law degree distribution. Pennock et al. \cite{pennock} introduced graphs with a similar degree distribution: the Discrete
Gaussian Exponential distribution. Amaral et al \cite{amral} observed exponential cut offs  in electric power-grid graph. Among all of these models, Barab\'asi-Albert model has been used widely and up to today particularly in the context of social networks. In this model a new vertex joins the network and it gets connected to each vertex $u$ with probability proportional to degree of $u$. 

\smallskip

Bianconi et al. \cite{Bianconi:1741450} observed that Barab\'asi-Albert graphs do not contain modular components like real networks do. They suggested two models to incorporate triadic closure and thus to create modularity. In the first (simpler) model, a vertex $u$ joins the network and it gets connected to a random vertex in the graph. For each of the next edges the end point is chosen through two rules: w.p. $p$ it is added to a randomly chosen vertex, and w.p. $1-p$  it is added to a randomly chosen vertex from neighbors of its neighbors. In the second model,  they introduce a \emph{fitness} factor for vertices. The probability of choosing end points for new edges will be proportional to the fitness factor.  Bianconi et al's model generalizes Barab\'asi-Albert's as setting $p=0$ and choosing  fitness factor equal to degree makes the two models equivalent. 

\smallskip

A similar line of research uses triadic closure for link prediction: 
Huang et al \cite{wwwrel, relhuang} study formation of closed triads. Further they  examine  affects of formation of the third tie in a triad  on the strength of the existing two ties \cite{huangKDD2018}.
 They incorporate  demographics and network topology for such predictions.
 Estrada et al \cite{relestra} use communicability distance functions to predict triadic closure. Romero and Kleinberg \cite{relKl} analyze a \emph{link copying}  regime which is similar to triadic closure. Despite a long list of similar work \cite{rel1, rel2,rel3, rel4}, none of them presents a theoretical analysis of triadic closure based solely on graph topology, its affect on rate of degree growth and its possible algorithmic implications. 

\medskip 

 We finish this introduction by highlighting our paper's contribution together with a road map to read the paper.

\subsection*{Our Contribution.}
\begin{enumerate}
\item In Section \ref{Sec:2}, we introduce the wedge picking model which  is a probabilistic generalization of c-closure, tailored to dynamic graphs. We estimate the parameters of our model  in a pre-processing phase. By doing so we  are also able to assert whether the evolution of input graph fits into the wedge picking model. 

\item In section \ref{sec:analys}, we analyze the wedge picking model and bound degree growth (as well as tri-degree defined as the number of triangles on a vertex) for each vertex in the network. 
\item In Section \ref{sec:restandrun}, we present our algorithmic contribution which is the introduction of the rest and run strategy. This strategy can be employed as a scheduling subroutine to boost the efficiency of other algorithms with negligible  effect on their accuracy.

In Subsection \ref{sec:algdens}, we show how to apply the rest and run strategy for approximating the densest subgraph (and tri-densest subgraph in the appendix). In particular, we incorporate the rest and run strategy in the algorithm of \cite{Epasto:2015}, and we show the approximation ratio of $4(1+\epsilon)$ for the modified algorithm. 
 
\item In Section \ref{sec:exp}, we complement our theoretical results by presenting experimental analysis. Our experimental analysis is twofold: (1) We compare the quality and the run time of our algorithms in comparison to previous work.  We observe that the rest and run strategy can effectively improve the speed by a factor of 1000 on large data sets while not degrading the quality of the solution. (2) 
We run our algorithms on two types of data sets: (a) real-world social graphs whose evolution is generated by our model, and  (b) dynamic real world graphs 
of ``Facebook wall posts'' and ``YouTube users'' (see Table \ref{tab:realdata}).
We compare the run time and accuracy of algorithms for synthetic and organic evolution and observe similar results. Thus, we demonstrate the applicability of our model on real world social graphs. 
 \end{enumerate}

\section{The model} \label{Sec:2}

In this section, we introduce our model. Let $G_t=\langle V,E_t \rangle$ be an undirected dynamic graph with a fixed vertex set $V$ and evolving edge set $E_t$ at time $t$.\footnote{For simplicity we assume that $G_t$ is undirected but our results can easily be modified to work for directed graphs. We defer this discussion to the full version of the paper.} The assumption of fixed vertex set is without loss of generality as we can assume that all the vertices exist in the graph from the initial state, and those vertices added to the graph in the course of evolution, initially have degree zero. Therefore, we focus on how $E_t$ evolves through time. \smallskip

Our model is a probabilistic  generalization of c-closure. Consequently for any disconnected pair of vertices $u$ and $v$ we assume the probability that they become connected is a function of their common neighborhood denoted by $d(u,v)$. 
Since  our experimental observations presented in Section \ref{sec:exp} show a linear dependence, we assume this probability is  $ad(u,v)+b$, for some $a$ and $b$ (learned in a pre-processing phase).

Edge removal occurs at a significantly lower rate compared to edge addition in social graphs. In our model, we choose uniform edge removal for deleting an edge, i.e., we pick a pair of connected vertices uniformly at random as a candidate for edge removal. Section \ref{sec:def} is devoted to a rigorous presentation of our model; this presentation has a combinatorial form which is simple and easy to analyze. 

\smallskip

Throughout this paper, we use the following notation: Let $G_{t} = (V, E_t)$ denote the graph with $n$ vertices at time step $t$, $G_{0}$ denote the initial graph, and  $N_{t}(v)$  the set of neighbors of vertex $v$ at time $t$, i.e., $N_{t}(v)=\{u\in V \vert (u,v)\in E_t\}$. We use $N_{t}(u,v)$ to denote the common neighborhood of vertices $u$ and $v$ at time $t$, i.e., $N_{t}(u,v)=N_{t}(u)\cap N_{t}(v)$. Let $d_t(v)$ denote the degree of vertex $v$ at time $t$ and $d_{t}(u,v)=\vert N_{t}(u,v)\vert$. To denote an induced degree, induced neighborhood or induced  common neighborhood in a subgraph $G'$, we mention the subgraph $G'$ explicitly and use the previous notation. We use the notation $\Gamma_t(v)$ to denote the number of wedges with one end point $v$, i.e., the number of paths of length two starting at $v$, and we use $\Gamma_t(V)$ or $\Gamma_t$ for short to denote the total number of wedges.

\subsection{The wedge picking model}\label{sec:def}
Let $p$, $q$ and $r$ be constants. Throw a fair coin to decide whether to follow rule (i) or (ii):

\begin{enumerate}[(i)]
    \item Pick a wedge, a path of length $2$, uniformly at random.   If the picked wedge is open (its endpoints are disconnected),  then \emph{close} it with probability  $p$ (place an edge between the endpoints). Note that one can pick a wedge uniformly at random by picking the midpoint vertex $v$ with probability proportional to ${{d_t(v)}\choose{2}}$, and then choosing the endpoints uniformly at random from $N_{t}(v)$. Furthermore by picking a wedge uniformly at random the probability of placing  an edge between a disconnected  pair of vertices $u$ and $v$ will be proportional to $d_t(u,v)$.
    
    \item Pick a pair of vertices uniformly at random if they are disconnected, connect them with probability $r$, if they are connected, disconnect them with probability $q$.
\end{enumerate}

So we can conclude that a disconnected pair $u$ and $v$ gets connected with probability $0.5(p \frac{d_t(u,v)}{\Gamma_t}+r{{n}\choose{2}}^{-1})$ which is linear in $d(u,v)$. Moreover, a connected pair becomes disconnected with probability $0.5(q \frac{d_t(u,v)}{\Gamma_t}+q{{n}\choose{2}}^{-1})$.

\subsection{ Parameters of the model}\label{sec:learnparams}

In this section, we explain how we learn whether or not a graph is evolving through the wedge picking model and if it does how to estimate the parameters $p$, $q$ and $r$. 
\medskip

\noindent\emph{\textbf{Learning $p$ and $r$.} }
In a pre-processing phase, we track edge additions between time $t$ and $t+\Delta$ where time $t$ is an initial time and parameter $\Delta$ will be determined later. Let $f(x)$ be the number of disconnected pairs with $x$ common neighbors in $G_t$ that become connected in $G_{t+\Delta}$, i.e., 
$f(x)=\vert \{(u,v): d_{t}(u,v)=x, (u,v)\in E_{t+\Delta} \setminus E_t\}\vert$.
For each $x \in \mathbb{Z}_+$, we also calculate $N_{t}(x)= \vert \{(u,v); d_t(u,v)=x, (u,v)\notin E_t\}\vert $. Using discussion in previous section, we know that the probability of edge addition is linear in $d_t(u,v)$; we assume it is $a d_t(u,v) + b$. We plot $x$ vs $\frac{f(x)}{N_t(x)}$, i.e., the probability of a pair of vertices with $x$ common neighbors becoming connected. Using linear regression we fit a line $ax + b $ to this plot and   conclude: 
\begin{equation}\label{eq:linear}
    \frac{f(x)}{N_t(x)}= a x + b.
\end{equation}

We use the magnitude of coefficient of determination as a certificate for guaranteeing a linear relation. For a coefficient of at least\footnote{In Algorithm \ref{algo:learn} this threshold can be set to any arbitrary value $c$ } $0.6$, we can conclude that the model follows our assumptions. Using the slope and intercept of the line then we can extract the values of $p$ and $r$ via Bayes's rule. 

\begin{thm}\label{thm:learnp}
If a graph evolves through the wedge picking model with parameters $p$ and $r$.
In Equation \ref{eq:linear},  we have:
$a=p\frac{M}{\Gamma_t}$ and $b= r \frac{M}{{{n}\choose{2}}}$, where $M$ is the probability of an edge being added to the graph in $[t,t+\Delta]$, i.e.,  $M=\left(p\frac{\Gamma^o_t}{\Gamma_t}+r(1-\frac{|E_t|}{{{n}\choose{2}}})\right)/\vert E_{t+\Delta}\setminus E_t\vert$ where $\Gamma^o_t$ is the number of open wedges at time $t$. 
\end{thm}

\begin{proof}
Let $\mathcal A$ be the event that an edge is added to the graph any time in $[t,t+\Delta]$. Using the definition of $M$ and the union bound we will have $\mathbb{P}({\mathcal A}) \approx M\cdot |E_{t+\Delta}\setminus E_t|$. Let $\mathcal B$ the event that an edge is added to a pair with $x$ common neighborhood. If the graph follows the wedge picking model, by Bayes rule we have:

$$
\begin{array}{ll}
    \mathbb{P}({\mathcal B}\vert{\mathcal A}) &=\mathbb{P} ({\mathcal A}\vert {\mathcal B}) \mathbb{P}({\mathcal B})/\mathbb{P}({\mathcal A}) \\
     & = p N_t(x)\left( \frac{x}{\Gamma_t}+r {{{n}\choose{2}}} ^{-1}\right)/\mathbb{P}({\mathcal A})\\
        & = p N_t(d)\left( \frac{x}{\Gamma_t}+r{{{n}\choose{2}}}^{-1}\right)/( M\cdot |E_{t+\Delta}\setminus E_t|).
\end{array}
$$

Thus, if we take $\Delta$  so that $N_t $ and $\Gamma^o_t/\Gamma_t$ (subsequently $M$) do not change significantly (see the while condition in Algorithm \ref{algo:learn}), and using the fact that $\mathbb{P}({\mathcal B}\vert{\mathcal A}) = \frac{f(x)}{|E_{t+\Delta}\setminus E_t|}$, we get

$$
\begin{array}{cl}
   f(x)/N_t(x)&=p\frac{M}{\Gamma_t}x+r \frac{
   M}{{{n}\choose{2}}}
\end{array}
$$

\end{proof}
\noindent\emph{\textbf{Learning $q$. }}
Our experimental analysis on real dynamic graphs suggest a very low rate of edge deletion compared to edge addition (see Section \ref{sec:exp}). In fact, in many of the dynamic graphs we have, edge deletion does not occur at all! Thus, instead of learning $q$'s value, we suffice to assert if $p/q\geq n$, which will be useful later on in the analysis of degree and tri-degree growth (see Lemma \ref{lemma:delta}).
The proof of the following theorem is straightforward and deferred to the full version. 
\begin{thm}\label{thm:learnq} For a graph evolving based on the wedge picking model we have:
\begin{equation}
   p/q\geq  \frac{\vert E_{t+\Delta} \setminus E_t\vert}{\vert E_t \setminus E_{t+\Delta} \vert} \cdot \frac{\Gamma_t}{\Gamma^o_t} \cdot\frac{\vert E_t\vert }{{n\choose 2}}
\end{equation}
\end{thm}

\section{Theoretical analysis of  degree  growth}\label{sec:analys}

In real social networks, number of edges and open and closed wedges change rapidly, and keeping track of this change, or calculating them at each time step can be expensive. In this Section we show how to theoretically bound the growth rate of degrees in order to avoid expensive bookkeeping.
\smallskip 

The rate of evolution of the graph at time $t$  is highly dependent on  the clustering coefficient of $G_t$ (defined as the ratio of total number of open and closed wedges). 
When the clustering coefficient is low, we expect a high rate of evolution; meaning not only
the number of edges, open and closed wedges are changing rapidly but also  the probabilities of adding or removing edges are. These rapid changes make design of algorithms a daunting task.
Nevertheless, in the next lemma (Lemma \ref{lemma:delta}) we show that these probabilities do not change significantly in limited time. 

\smallskip 
Let $\Gamma^k_t(u)$ be defined as the number of paths of length $k$ at time $t$ having one end at $u$, and  $\Gamma^k_t(u,v)$ the number of paths  between $u$ and $v$ at time $t$ and of length $k$. For example, $\Gamma^1_t(u,v)=1$ if $(u,v)\in E_t$ and $0$ otherwise, $\Gamma^2_t(u,v)=d_t(u,v)$, and $\Gamma^1_t(u)=d_{t}(u)$. Since $\Gamma^2$ is the most commonly used notation in our paper, we use $\Gamma$ to refer to $\Gamma^2$. Note that  at any time step, closing or  a wedge residing on a path would decrease the length of the path. Thus,  $\Gamma^i_{t+\delta}(u,v)$ will depend on $\Gamma_t^{i+1}(u,v),\dots,\Gamma^{i+\delta}_t(u,v)$. Moreover, for an arbitrary $k$ and  any $i<k$ by connecting two paths $\gamma_1\in \Gamma^i_t(u)$ and  $\gamma_2\in\Gamma^{k-1-i}_t(v)$ we can create a path of length $k$. We use $d_{max}$ to denote the maximum degree of the graph $G_t$. We formalize the above observations, in the following lemma: 
\begin{lemma}\label{lemma:delta}
For any $u,v\in V$ and $\delta \leq e \Gamma_t/(2pd_{max})$, $q =o( p/ d_{max})$, and $r\leq (n^2/\Gamma_t)p$ with probability $1-o(1)$ we have:
\begin{equation}\label{eq:lemmadelta}
\mathbb{E}[d_{t+\delta}(u,v)- d_{t}(u,v) ]\leq 4\delta \frac{d_t(u)d_t(v)}{\Gamma_t}.
\end{equation}
\end{lemma}

\begin{proof}
For simplicity we first assume that we have no edge deletion. We will take care of the deletions later.

To warm up, consider the probability $\mathbb{P}\left[d_{t+1}(u,v)- d_{t}(u,v)> 0\right]$ there are two ways to add to  $d_t(u,v)$: (1) by closing only one wedge  placed on a path of length three between $u$ and $v$, and there are two ways of doing so (rule (i)), (2) by adding an edge between $w\in N_t(v)$ and $u$ or vice versa (rule (ii)).

Thus, $\mathbb{P}\left[(d_{t+1}(u,v)- d_{t}(u,v))> 0\right]\leq 2\Gamma^{3}_{t}(u,v)p/\Gamma_t+ (d_t(u)+d_{t}(v))r/n^2$, equivalently:

$\mathbb{E}[d_{t+1}(u,v)- d_{t}(u,v)]\leq 2\Gamma^{3}_{t}(u,v)p/\Gamma_t+(d_t(u)+d_{t}(v))r/n^2$.

\medskip

We now consider the general case of $\delta$ steps:

\begin{equation*}
 \mathbb{E}[d_{t+\delta}(u,v)- d_{t}(u,v)]= \sum_{i=1}^{\delta} \mathbb{E}[d_{t+i}(u,v)- d_{t+i-1}(u,v)]
=  \sum_{i=1}^{\delta}\mathbb{P}[ \text{creating a wedge at  }t+i].
\end{equation*}

\begin{figure}[h!]
\begin{center}
\setlength{\unitlength}{1cm}

\begin{picture}(2,3.2)
{\put(1.3,3.2){\circle*{0.15}}}
\put(-0.8,2.2){\line(-2,3){0.5}}
{\color{red}
\put(0.2,2.2){\line(-2,1){1.5}}}
{\color{blue}
\put(0.2,2.2){\line(1,1){1}}}
\put(0.7,3.2){$w$}
{
\qbezier(1.2,3.2)(2.6,3.5)(2.6,3)}
{\color{green}
\qbezier(0,2.2)(2.5,3.5)(2.5,3)}

\put(0,2.2){\circle*{0.15}}
\put(2,2.2){\circle*{0.15}}
\put(1,2.2){\circle*{0.15}}
\put(-1,2.2){\circle*{0.15}}
\put(2,2.2){\line(2,3){0.5}}

\put(-1.8,3){$u$}
\put(-1.5,3){\circle*{0.15}}

\put(2.5,3){\circle*{0.15}}
\put(2.6,3){$v$}
\put(0,2.2){\line(1,0){1}}
\put(1,2.2){\line(1,0){1}}
\put(-1,2.2){\line(1,0){1}}

\put(-0.8,0.2){\line(-2,3){0.5}}
{\color{red}
\put(0.2,0.2){\line(-2,1){1.5}}}
{\color{blue}
\qbezier(0.1,0.2)(0.5,1)(2,0.2)}
{\color{green}
\qbezier(2,0.2)(1,1)(-1.5,1)}

\put(0,0.2){\circle*{0.15}}
\put(2,0.2){\circle*{0.15}}
\put(1,0.2){\circle*{0.15}}
\put(-1,0.2){\circle*{0.15}}
\put(2,0.2){\line(2,3){0.5}}

\put(-1.8,1){$u$}
\put(-1.5,1){\circle*{0.15}}

\put(2.5,1){\circle*{0.15}}
\put(2.6,1){$v$}
\put(0,0.2){\line(1,0){1}}
\put(1,0.2){\line(1,0){1}}
\put(-1,0.2){\line(1,0){1}}
\end{picture}
\caption{After three steps, a path of length 5 between $u$ and $v$ can turn to a wedge between them applying  rule (i) (below picture); or alternatively  two disjoint paths can turn to a wedge between $u$ and $v$ by applying rules (i) and (ii) (top picture). The order of edges being added is: red, blue, green.}
\end{center}
\end{figure}
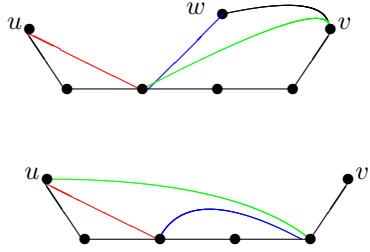
Note that creating a wedge at time $t+j$ is possible by  closing wedges (applying rule (i)) on a path which has length $j+2$ at time $t$. Having rule (ii) in our hand, we can create new paths of length $j$ at time $t$ by placing an edge between a path of length $i$ in $\Gamma^i_{t}(u)$ and a path of length $j-1-i$ in $\Gamma^{j-1-i}_{t}(v)$. Thus the original expectation is bounded by: 
\begin{equation*}
     \sum_{i=1}^{\delta}\frac{d_{t+i}(v)+d_{t+i}(u)r}{n^2}+ 2\Gamma^{3}_{t+i}(u,v)  \frac{p}{{
     \Gamma_{t+i}}}.
\end{equation*}

We now bound the first and second term of the sum:
\begin{equation}
     \sum_{i=1}^{\delta}\frac{d_i(v)+d_i(u)r}{n^2} \leq \frac{\delta r(d_t(v)+d_t(u))(1+\delta/n)}{n^2}
\end{equation}

\begin{dmath}\label{eq:3}
    \Gamma^3_{t+i}(u,v)\leq \sum_{j=1}^i d_{max}^j \left(d_t(u)d_t(v)\frac{ p^j \kappa(2+j)}{ \Gamma_t^j}+
    (d_t(v)+d_t(u))\frac{ 2 r^j}{ n^{2j}}\right)
\end{dmath}

, where $\kappa(2+j)$  is the number of ways a path of length $2+j$ can be turned into one wedge by closing $j$ wedges consecutively. And $\Gamma_t$ is the total number of wedges at time $t$.
 
 It is an easy counting problem to see $\kappa(j+2)=(j+1)!$. Note that  for any $t$, $\Gamma_t\geq d_{\max}^2$, and for each $j$, $\Gamma_t^{2+j}(u,v)\leq d_t(u)d_t(v)d_{\max}^j$, and $\Gamma_t^{j}(u)\leq d_t(u)d_{\max}^j$ where $d_{\max}$ is maximum degree at time $t$.   Note that
  $p^j j!/d_{\max}^j\leq (p
  \delta/d_{\max})^j$. Under the assumption that  $\delta \leq e\Gamma_t/2pd_{max}$ we can simplify Equation \ref{eq:3} to get:

 \begin{equation}
 \begin{array}{ll}
 \sum_{i=1}^{\delta}\Gamma_{t+i}^3(u,v)& \leq d_t(u)d_t(v) \sum_{i=1}^{\delta}\sum_{j=1}^i (j+1) k^j  ,\quad k<1/2\\
 &\leq (4\delta) d_t(u)d_t(v).
 \end{array}
 \end{equation}

Thus, as long as $r\leq (n^2/\Gamma_t)p$ (consistent with our experiments) we have:

 \begin{equation}\label{equ:delta}
 \mathbb{E}[(d_{t+\delta}(u,v)- d_{t}(u,v))]\leq \frac{d_t(u)d_t(v)}{\Gamma_t}  (4\delta).
 \end{equation}
 
 We now show that our initial assumption is not restrictive us as long as $q =o( p/ d_{max})$. In fact we can show that with probability $1-o(1)$, no edge deletion will occur in time $\delta\leq e d_{max}/p$:

\begin{equation*}
\begin{array}{ll}
\mathbb{P}( \text{at least one edge is deleted})\leq &\sum_{t=1}^{\delta}q\sum_{u,v} 1/n^2\\\mathbb{P}( \text{at least one edge is deleted})\leq &\delta p/\omega(d_{max})\leq o(1).
\end{array}
\end{equation*}
\end{proof}

\medskip

Having Lemma \ref{lemma:delta}, for each $v\in V$, we can now bound the rate of its degree or tri-degree (defined later) growth:

\medskip

\noindent\emph{\textbf{Degree of a vertex $u\in V$.}}  
\medskip 

Consider and arbitrary vertex $v\in V$, the following lemma bounds the rate of the growth of $u$'s degree in the wedge picking model: 

\begin{lemma}\label{lemma:degree} Let $\delta \leq e \Gamma_t/2pd_{max}$ and $q =o(d^{-1}_{max})p$,
for any $u\in V_t$:
\begin{equation}\label{eq:deg1}
\mathbb{E}[(d_{t+\delta}(u)- d_{t}(u)) ]\leq \delta \left(2p\frac{ \Gamma_{t}(u)}{\Gamma_t} +  r\frac{1}{n}\right).
\end{equation}

, where $\Gamma_{t}(u)$ is the number of wedges having one endpoint equal to $v$, i.e. $ \Gamma_{t}(u)=\sum_{v\in V}  d_t(u,v)$. 
\end{lemma}
\begin{proof}

 Let $\vert E_t\vert=m_t$, we first employ the Cauchy Schwarz inequality to show that $ \Gamma_t/m_t\geq (1/2)(d_{avg}-1)$: $\Gamma_t=\sum_{v\in V} (d_t^2(v)-d_t(v))/2, \sum_{v\in V} d_t^2(v)\geq \left(\sum_{v\in V}d_t(v)\right)^2/n_t $ ; thus, $2\Gamma_t+m_t\geq m_t^2/n_t$.
$$
\begin{array}{lll}
\mathbb{E}[(d_{t+\delta}(u)- d_{t}(u)) ]&\leq \sum_{i=1}^{\delta}\mathbb{E}[(d_{t+i}(u)- d_{t+i-1}(u)) ]\\ 
&\leq \sum_{i=1}^{\delta} \sum_{u\in V} d_{t+i-1}(u,v)\frac{p}{\Gamma_t}+\frac{\delta\cdot r}{n}.\\
\end{array}
$$
Employing Lemma \ref{lemma:delta}:

$$
\begin{array}{lll}
\sum_{i=1}^{\delta} \sum_{u\in V} & d_{t+i-1}(u,v)\\
&\leq  \sum_{i=1}^{\delta} \sum_{u\in V}  \left(d_t(u,v)+\frac{4(i-1)pd(u)d(v)}{\Gamma_t}\right)\\ 
&\leq \Gamma_{t}(v)\delta +\sum_{i=1}^{\delta} \frac{4(i-1)pd(v) m_t}{\Gamma_t}.
\end{array}
$$
Having  $d_t(v)\leq \Gamma_{t}(v)$  and  $\frac {2 \Gamma_t}{m_t}\geq d_{avg}-1$, we get
$$
\begin{array}{lll}
\sum_{i=1}^{\delta} \sum_{u\in V} & d_{t+i-1}(u,v)\\
&\leq (p/\Gamma_t)\Gamma_{t}(v)\delta(1+(\frac{\delta}{2})(\frac{4pd(v)m_t}{\Gamma_t\Gamma_{t}(v)})) \\
&\leq (p/\Gamma_t)\Gamma_{t}(v)\delta(1+(\frac{\delta}{2})(\frac{2p}{d_{avg}}))\\ & \leq 
\frac{\delta \Gamma_{t}(u)p}{\Gamma_t} \left(1+\frac{2p\delta}{\Gamma_t d_{avg}}\right)
\\
& \leq 
\frac{\delta \Gamma_{t}(u)p}{\Gamma_t} (2), \text{ for }\delta \leq \Gamma_t/ d_{\max}p.
\end{array}
$$

\end{proof}

\smallskip

\noindent\emph{\textbf{Tri-degree of a vertex $u\in V$.}}

\medskip 

For each vertex $u$, and time stamp $t$, let  tri-degree of $u$ at time $t$, denoted by $D_t(u)$, be the number of triangles having $u$ as a vertex. We have the following lemma whose proof is similar to proof of Lemma \ref{lemma:degree} and omitted:
\medskip

\begin{lemma}\label{lemma:tri}

For any $u\in V$ and $\delta \leq e d_{max}/2p$  and $q =o(d^{-1}_{max})p$,
\begin{dmath}
\mathbb{E}[(D_{t+\delta}(u)- D_{t}(u)) ]\leq \frac{\delta \zeta^2_{t}(v)p}{\Gamma_t} \left(1+\frac{8 \delta p}{d_{avg}}+\frac{6\delta^2p^2}{\Gamma_t }\right)+\frac{d^2_t(v)
(\delta+\delta^2/n)}{n^2}
\end{dmath}
, where $ \zeta^2_{t}(v)=\sum_{u\in V}  (d_t(u,v))^2.$

\end{lemma}

\begin{remark}
The growth rate of other notions of degree  can be bounded similarly. These notions have been used in the literature  to define various density measures \cite{Babis2013, Gionis:2015:DSD:2783258.2789987}.
\end{remark}
\section{Algorithmic implications: The rest and run strategy}\label{sec:restandrun}
In the dynamic graph literature, algorithms are often designed to update output solution after every single edge update; this is essential when edge updates arrive in an arbitrary (or adversarial) manner.   
Consider an algorithm designed for dynamic graphs which despite having low amortized cost, occasionally performs expensive operations in a row. Batching edge updates improve the run-time of these algorithms significantly. However, when not designed carefully, batching can harm the accuracy of algorithms. 
\smallskip 

In this section, we utilize our previous results  to devise a speed up subroutine based on batching for some algorithms without harming their accuracy. More precisely, for a given graph $G_t$, we identify time $\Delta_t$ where a previously calculated solution still serves as a good approximation for any graph $G_{t'}$; $t'\in[t,t+\Delta_t]$. Our subroutine declares $[t, t+\Delta_t]$ as a ``rest'' interval. During rest intervals  the algorithm only  collects the changes in the graph and does not perform expensive operations. In the ``run'' intervals, the algorithm has its normal performance except since the changes are received in batches it saves  redundant operations; in the experimental section we observe significant speed up due to this simple trick while maintaining the accuracy compared to the optimal solution.

\smallskip 

Given a problem of finding a subgraph $S$ in $G$ satisfying a certain property, e.g., the problem of densest subgraph discovery, the  question we seek to answer is:  

\vspace{0.8cm}

\fbox{
 \parbox{15 cm}{\center{
\emph{\small Given an approximate solution $S$ for $G_t$, what is the largest $\Delta_t$ such that with 
high probability, for any $t'\leq t+\Delta_t$, $S$ is a still ``good'' approximation for $G_{t'}$.\vspace{.25cm}}
}}
}

\vspace{0.8cm}

Note that if an algorithm's functionality is based on the degrees (tri-degree) of vertices of the input graph (or its degree in other induced subgraphs), Lemma \ref{lemma:degree}  (Lemma \ref{lemma:tri}) provides us with a value of $\Delta_t$:

\smallskip
Let $C_t(v)= 2p\frac{ \Gamma_{t}(v)}{\Gamma_t} +  r\frac{1}{n}$, and $\beta_t(1+\epsilon)$ be the change in degrees that the vertices in  $S$ afford to undergo while maintaining the quality within  $(1+\epsilon)$ precision. We define  $\Delta_t$ as follows:
\begin{equation}\label{Eq:tau}
 \Delta_t\dot{=}\min_{v\in V}\tau_{v},\text{ where } \tau_{v}\dot{=}\min \left \{ \frac{e \Gamma_t}{2pd_{max}}, \frac{\beta_t
(1+\epsilon)}{C_t(v)}\right\}.
\end{equation}

Employing Lemma \ref{lemma:degree}, we conclude that using the above value for $\Delta_t$, for any vertex $u$ (in particular the vertices in solution $S$) and $\delta\in [t,t+\Delta_t]$ we will have:

\[
\mathbb{E}[d_{\delta}(u)- d_{t}(u) ]\leq (1+\epsilon) \beta_t.
\]

We formalize this observation in the next section for finding the densest subgraph in a dynamic graph. In section \ref{tridense} we show how to  generalize these results for finding tri-densest subgraph with trivial modifications.

\subsection{The densest subgraph problem}\label{sec:algdens}
In this section, to demonstrate the applicability of the rest and run strategy, we apply it to the problem of finding the densest subgraph. We present a similar algorithm for finding tri-densest subgraph (a subgraph with maximum density of triangles) in the Section \ref{tridense}.  We believe that our work provides initiatives to design of algorithms for finding approximations to the \emph{$k$-clique densest subgraph} problem \cite{Tsourakakis:2015:KDS}.

\smallskip 

Finding maximum cliques in a graph and even approximating it within a reasonable factor is NP hard \cite{Hastadcliqueis}. Thus, researchers often define density of a subgraph as its  induced average degree. In the literature, \emph{the densest subgraph problem} is  discovery of subgraphs with maximum average degree (formal definition in Equation \ref{def:densest}). This problem is  polynomially tractable using flow computations \cite{Goldberg:1984:FMD:894477}, and  a greedy algorithm providing a $2$ approximation \cite{Charikar:2000}, also known as \emph{peeling} is used in practice.  Peeling is a key component of most of  the state of art algorithms  in various massive data settings. For example:   (1) in data streaming model where the network is processed in passes and the goal is to obtain a solution with minimal number of passes and low memory Bahmani et al propose an algorithm based on peeling \cite{Bahmani:2012}, (2) in the MapReduce model where the data  is distributed between different machines a number of algorithms have been suggested based on peeling \cite{Bahmani:2012,Bhattacharya:2015, Epasto:2015} and (3) the dynamic setting in which the graph has an evolving edge or vertex set, algorithms based on peeling have been shown to be effective \cite{Bhattacharya:2015, Epasto:2015}.  A summary of the results  based on peeling is presented in Table \ref{tab:sum}.  It is important to mention here that peeling is not the only main strategy to find densest subgraph and in fact several foremost algorithms maintaining sketches or sparsifiers have been suggested in data streaming and dynamic graph models \cite{Sawlani2019, Esfandiari2015, Mcgregor2015, Mitzenmacher2015,Epasto:2015}. We also remark that we do not aim to beat any of these algorithms theoretically. Our solution is to boost their runtime in practice by employing the \emph{rest and run strategy}.

\begin{table}[hbt]
\small
\begin{center}
\begin{tabular}{|c|c|c|c|c|}
\hline
Authors/year& Model&  Complexity&  Quality& Problem \\
\hline
\hline
Bahmani et al, & Streaming,&\#passes:&&densest\\
2012 \cite{Bahmani:2012}& MR&$\log_{1+\epsilon}n$&$2(1+\epsilon)$&subgraph\\
\hline
Epasto et al, & {Dynamic,}&  ACPU &&densest\\
 2015 \cite{Epasto:2015}& {MR}&$\log_{1+\epsilon} n$ &$2(1+\epsilon)$&subgraph\\
\hline
{Bhattacharya  }  &{Dynamic,}&ACPU&&densest\\
et al, 2015 \cite{Bhattacharya:2015} &{Streaming, MR}& $\log_{1+\epsilon}$ &$4(1+\epsilon)$&subgraph \\
\hline
{Tsourakakis, }& Dynamic,&\#passes:&&tridensest\\
{ 2014 \cite{Tsourakakis14anovel}}& MR&$\log_{1+\epsilon} n$&$3(1+\epsilon)$&subgraph\\
\hline
\end{tabular}
\caption{A summary of the previous work on densest and tridensest subgraph discovery based on ``peeling''. MR stands for MapReduce and ACPU stands for amortized cost per update. }
\label{tab:sum}
\end{center}
\end{table}
\smallskip 

Our proposed solution for a faster densest subgraph discovery algorithm is to equip \emph{ peeling}  with the rest and run strategy.
We show that while maintaining  $4(1+\epsilon)$ approximation  for the densest subgraph problem our algorithm has significant improvement in its running time compared to its counterparts.

\medskip

Our algorithm follow the following three steps: (1) We learn the parameters of the wedge picking model in a pre-processing phase (2) We find rest intervals after each run of the algorithm (3) we optimise the batch of the updated edges to feed them into the peeling algorithm. 

\smallskip

Let the density of graph $G=(V,E)$, denoted by $\rho(G)$, be the ratio of the number of edges in $G$ to the number of vertices of $G$, i.e. $\rho(G)=\vert E \vert/ \vert V\vert$.
In the densest subgraph problem, our goal is to find a subgraph with maximum density. Let  $\beta^{*}(G)$ denote the maximum density among all subgraphs of $G$, and let $densest(G)$ be a subgraph $G$ with maximum density, so
\begin{equation}\label{def:densest}
\beta^{*}(G)= \max_{G'\subseteq G}\rho(G'), ~~ densest(G)= \underset{G'\subseteq G}{\rm argmax}~\rho(G').  
\end{equation}

In the dynamic graph setting algorithms of Epasto et al \cite{Epasto:2015} and Bhattacharya et al \cite{Bhattacharya:2015}, both based on the peeling strategy, present constant factor approximation algorithms, with amortized cost $poly(\log n)$ per edge addition/deletion. Despite having a low amortized cost, the above algorithms fail to output a solution in reasonable time when the size of a network ($n$) is huge ($n\simeq 1M$) while the edges are evolving at a high rate; see Section \ref{sec:exp}.

\smallskip 

In this section, after  briefly presenting the peeling strategy, we introduce our algorithm. In particular, our emphasis is on Epasto et al.'s algorithm \cite{Epasto:2015} which adopts peeling to dynamic graphs, and proves an amortized cost of $\log(n)$ per edge update, while maintaining a $2(1+\epsilon)$ approximate solution. We assume that the reader is familiar with their results and in a few occasions we rely on their work in  our proofs (for more information see \cite{Epasto:2015}). 

\smallskip

\noindent\emph{\textbf{Peeling.}}
Let $\beta^{*}$ be the optimum density as defined in Equation \ref{def:densest}. Partition the vertices of graph $G=\langle V,E \rangle$ into layers $S_0,S_1,\dots , S_k$ as follows to find a subgraph with density at least $\beta^{*}/2(1+\epsilon)$:

$S_0=V$, while $S_{r+1}\neq \emptyset$, derive $S_{r+1}\subseteq S_r$ by removing vertices whose degree in the induced graph $S_{r}$ is less than $2 \beta(1+\epsilon)$.

\smallskip

When the graph is static, a multiplicative $(1+\epsilon)$ approximation of $\beta^{*}$ can be approximated by iterating over all $\log_{1+\epsilon}(n)$ values.  For the dynamic case Epasto et al. suggest maintaining the following invariant if possible, and if not possible rebuilding the layers by iterating over all the $\log_{1+\epsilon}(n)$ values for $\beta$ as in the static case:
\medskip 

\textbf{Invariant:} For a given density $\beta$, parameter $\epsilon$, and current graph $G_t= (V_t, E)$, the nodes of the graph are organized in layers $S_0, S_1, \cdots , S_k$ for $k = \lceil \log_{1+\epsilon}{n} \rceil$, where  $S_0=V_t$, $S_{r+1} \subseteq S_r$ is obtained from removing vertices with degree less than $2\beta(1+\epsilon)$ from $S_r$, and $S_k = \emptyset$.

\
\begin{algorithm}
\begin{algorithmic}[\ref{algo:peel}]
\REQUIRE  threshold $c$ , A series of edge updates: $ \{\langle(u,v)_{t_i},{\chi}_{t_i}\rangle\}_{i=0}^{\infty}$, where $(u,v)_{t_i}$s are pairs of vertices where the update takes place at time $t_i$ and $\chi_{t_i}$ is either addition or removal of an edge. 
\ENSURE  Layers $S_0,S_1, S_2,\dots , S_l$, 
\STATE At time $t_0$: Run Epasto et all algorithm and find $S_1, S_2, \dots , S_l$.
\STATE At time $t_0$, run Algorithm 2 and find $p$, $q$, $r$ and $\mathcal R$.
\IF{${\mathcal R}\geq c$} 
\STATE \textbf{continue}
\ENDIF
\REPEAT
\STATE having $p,q$ and $r$, find $\Delta_t$
\WHILE{$t_i\leq t_0+\Delta_t$}
\STATE Update the graph and do nothing, let $U$ be the set of updated vertices 
\ENDWHILE
\STATE Push all the vertices in $U$ in the stack and run algorithm of Epasto et al. to maintain $S_0,S_1, S_2,\dots , S_l$.
\STATE $t_0=t_0+\Delta_t$.
\UNTIL{$\infty$}
\end{algorithmic}
\caption{Rest and run densest sub graph.}
\label{algo:peel}

\end{algorithm}

\medskip

\begin{algorithm}
\begin{algorithmic}[]
\REQUIRE An initial graph $G_{0}$, a series of edge updates: $ \{\langle(u,v)_{t_i},{\chi}_{t_i}\rangle\}_{i=0}^\infty$, where $(u,v)_{t_i}$s are pairs of vertices where the update takes place at time $t_i$ and $\chi_{t_i}$ is either addition or removal of an edge. 
\ENSURE  $p, q , r$ (as defined in Section \ref{sec:def}), and linear regression's coefficient of determination.
\STATE $t\gets t_0$.
\STATE for each $0\leq x \leq n$, let $f(x)=0$ and $f'(x)=0$.
\STATE $A\gets 0, D\gets 0$  respectively the number of added/deleted edges.

\WHILE{
$\frac{\Gamma^o_t}{\Gamma_t} \in [\frac{1}{1+\epsilon}, 1+ \epsilon] \frac{\Gamma^o_0}{\Gamma_0} $, $\forall x: N_t(x) \in [\frac{1}{1+\epsilon}, 1+ \epsilon] N_0(x),  t \leq \vert E_{0}\vert$}
\STATE $x=d_{t}(u,v)$
\STATE update $f(x), A, D$
\STATE Update $\Gamma_t^o$, $\Gamma_t$, and $N_t$
\STATE $t \gets t + 1$
\ENDWHILE
\STATE $a,b \gets$ applying linear regression to $f(x)/N_{0}(x) = a x +b$.

\STATE Using Theorems \ref{thm:learnp} and \ref{thm:learnq}  find $p,r$, $q$, and the coefficient of determination $\cal R$. 
\RETURN $p, q, r, {\cal R}$
\end{algorithmic}
\caption{Learning Phase.}
\label{algo:learn}
\end{algorithm}
\noindent\emph{\textbf{Our Algorithm: rest and run + Peeling}}.

\medskip

The learning phase of our algorithm which was explained in full detail in Section \ref{sec:learnparams} is presented as Algorithm \ref{algo:learn}. The rest and run intervals are found in the UpdatePeeling, Algorithm \ref{algo:peel}.
The correctness of Algorithm \ref{algo:peel} algorithm can be easily concluded from correctness of Epasto et al's algorithm \cite{Epasto:2015} and the following lemma:

\begin{lemma}\label{lemma:chernoff}For an arbitrary vertex $v\in V$, let $d'(v)$ be degree of $v$ in layer $S_l(v)$ at time $t+\tau(v)$. For $\beta_O\geq c \log(n)$ for constant $c>1$ with probability $1-o(1)$ for all  vertices $v$ in $G\setminus S_l$, we have:
\begin{equation}
d'(v)\leq 4(1+\epsilon) \beta_O.
\end{equation}
\end{lemma}
\begin{proof}
For each $v$ consider a sequence of random variables $X_1,X_2,\dots X_{\tau(v)}$, $X_i$ indicating degree incrementation of $v$ at time $t+i$. Note that using the values of $\tau$ defined in Equation \ref{Eq:tau} we have $\mathbb{E}[d_{t+\tau(v)}(v)-d_{t}(v)]\leq \beta(1+\epsilon)$. Applying the Chernoff bound we have: $\mathbb{P}\left[d_{t+\tau(v)}(v)-d_{t}(v)\geq 2 \beta(1+\epsilon)\right] \leq  e^{-(1+\epsilon)\beta}\leq \omega(1/n)$. Summing over all vertices in $G\setminus S_l$ and using the union bound we obtain: 
\begin{eqnarray*}
\mathbb{P}\left[\exists v\in G\setminus S_l; ~d_{t+\tau(v)}(v)-d_{t}(v)\geq 
2 \beta(1+\epsilon)\right]\\ 
\leq (n-\log(n))\omega(1/n)=o(1).
\end{eqnarray*}
\end{proof}

The following remark is concluded from Theorem \cite{Epasto:2015} and Lemma \ref{lemma:chernoff}:

\smallskip

\emph{Remark 1.} Algorithm \ref{algo:peel} maintains a $4(1+\epsilon)$ approximation to the densest subgraph problem. 
And if $A$ is the total number of  edges added to the graph in between any two consequent calls of Algorithm \ref{algo:peel} and $U$ is the total number of affected vertices, then since any batch of edges that are added to a single vertex in the graph are processed only once, this will cause the runtime of our algorithm to drop by a ratio of $\nicefrac{A}{U}$ compared to that of Epasto et al \cite{Epasto:2015}.

\medskip

In the next section,
 using Lemma \ref{lemma:tri} we employ the rest and run strategy to find tri-densest subgraphs (subgraphs with maximum average of tri-degree) within $5(1+\epsilon)$ approximation factor.

\subsection{Finding Tri-densest Subgraph}\label{tridense}

Algorithms for densest subgraph may output  subgraphs  with small clustering coefficient and large diameter. Hence studying other similar problems have been suggested in the literature \cite{Babis2013,Tsourakakis:2015:KDS,Tsourakakis14anovel}: one example is searching for subgraphs which have maximal density of triangles i.e.  discovery of \emph{tri-densest subgraphs}. Note that a subgraph with a higher density of triangles  has a higher clustering coefficient and is a better representative of triadic closure.

\medskip  
We now proceed to present out solution to this problem in the wedge picking model:

\smallskip

Let the tri-density of a graph $G=(V,E)$, denoted by $\xi(G)$, be the ratio of the number of triangles in $G$ to the number of vertices of $G$, i.e. $\xi(G)=\vert W^c(G') \vert/ 3\vert  V'\vert$, where  $W^c(G')$ is its number of closed wedges (Note that each triangle counts as three closed wedges). In the tridensest subgraph problem, our goal is to find a subgraph with maximum tridensity. Let $\alpha_{O}(G)$ denote the maximum tridensity among all subgraphs of $G$, and let $tridensest(G)$ be a subgraph with maximum density, so 
$$\alpha_{O}(G)= \max_{G'\subseteq G}\xi(G'), ~~ tridensest(G)= \underset{G'\subseteq G}{\rm argmax}~\xi(G').$$
The problem has a $3(1+\epsilon)$ solution when updates are adversarial: maintain  peeling layers by removing vertices with tridegree less than $3(1+\epsilon)\alpha$ consecutively \cite{Tsourakakis14anovel}. 
Following the  strategy of rest and run, at each time $t$ we run the algorithm $UpdateTriPeeling$ of  \cite{Tsourakakis14anovel}, calculate $\Delta_t$. Then, we wait until  $t+\Delta_t$ for the next call to $UpdateTriPeeling$.

\begin{algorithm}\label{algomatrix2}
\begin{algorithmic}[\ref{algomatrix2}]

\REQUIRE  Layers $S_0,S_1, S_2,\dots , S_l$ , $T$, most recent reading time $t$.
\ENSURE  $5(1+\epsilon)$ approximation to densest subgraph, most recent reading time $t'$
\STATE Let $\Delta_t=\min_{v}\{\tau_{v} \in T\}$;
\STATE Read $\Delta$ edge updates: $U=\{(u_1,v_1), (u_2,v_2),\dots, (u_{\Delta_t},v_{\Delta_t}) \}$;
\STATE $E=E\cup U$
\STATE $G_{t+\Delta_t}=\langle V, E \rangle$
\STATE Update $N=\{N(v)\}_{v\in V}$ and $T=\{\tau(v)\}_{v\in V}$ correspondingly;
\STATE $UpdateTriPeeling(G_{t+\Delta_t},U)$;
\RETURN $S_l$, $t'$
\end{algorithmic}
\caption{Update (tri-Densest).}
\end{algorithm}\label{algo:2}

The following Theorem proves the correctness of Algorithm \ref{algo:2}, and its proof is identical to proof of Lemma \ref{lemma:chernoff} and thus omitted. 

\begin{thm}\label{thm:last} In the wedge picking model, and for $\alpha_{O}\geq c\log(n)$; $c>1$ a constant, there is an algorithm  which maintains a $5(1+\epsilon)$ approximation of tridensest subgraph problem with probability $1-o(1)$ and low amortized cost.  
\end{thm}

\section{Experiments}\label{sec:exp}

\subsection{Asserting the validity of wedge picking model assumptions}\label{sec:exp2}

In order to assert whether or not the probability of adding an edge between any pair of disconnected vertices is linear with respect to the size of their common neighborhood, we run algorithm \ref{algo:learn} the graphs of Table \ref{tab:realdata}. Plotting $f(d)/n(d)$ over $d$, and running linear regression, we observe two main general behaviours for these graphs:
(a) coefficient of determination is large. i.e. higher than $0.6$; see e.g. Figure \ref{fig:learn}a. (b) coefficient of determination is reasonable. i.e. between $0.2$ and $0.5$; see e.g. Figure \ref{fig:learn}b. For most of the graphs with coefficient of determination between $0.2$ and $0.5$, the plot shows an interesting branching behaviour: the dependence  is linear but it is in fact more that one  line; see e.g. Figure \ref{fig:learn}b. 

\smallskip

In both of the above cases we conclude linearity.

\smallskip

\begin{figure}[!hbt]
    \centering
        \begin{subfigure}[b]{ 0.45 \textwidth}
            \centering
            \includegraphics[width=\textwidth]{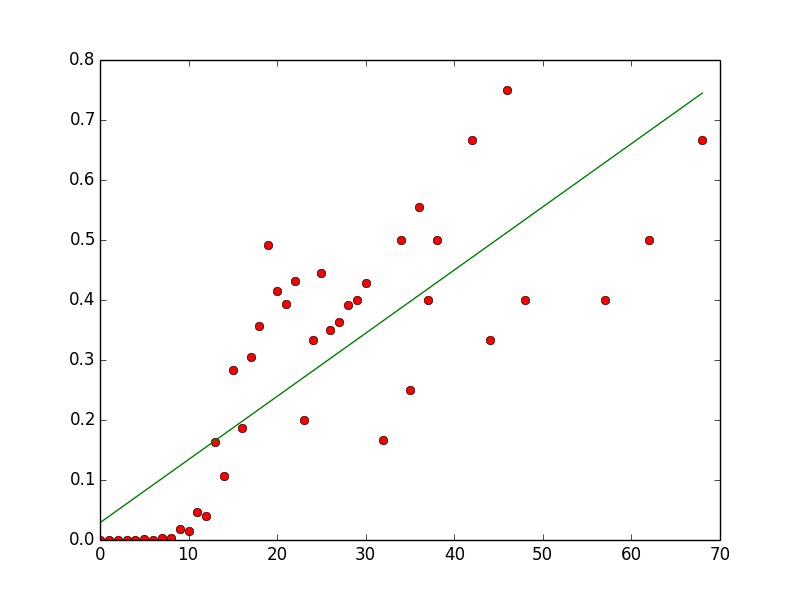}
            \caption{dataset: Digg Friends, Slope: 0.01, intercept: 0.02, coefficient of determination: 0.66. The high value of the coefficient of determination certifies linear relation.}
        \end{subfigure}
        \hfill
        \begin{subfigure}[b]{0.45 \textwidth}
           \centering
           \includegraphics[width=\textwidth]{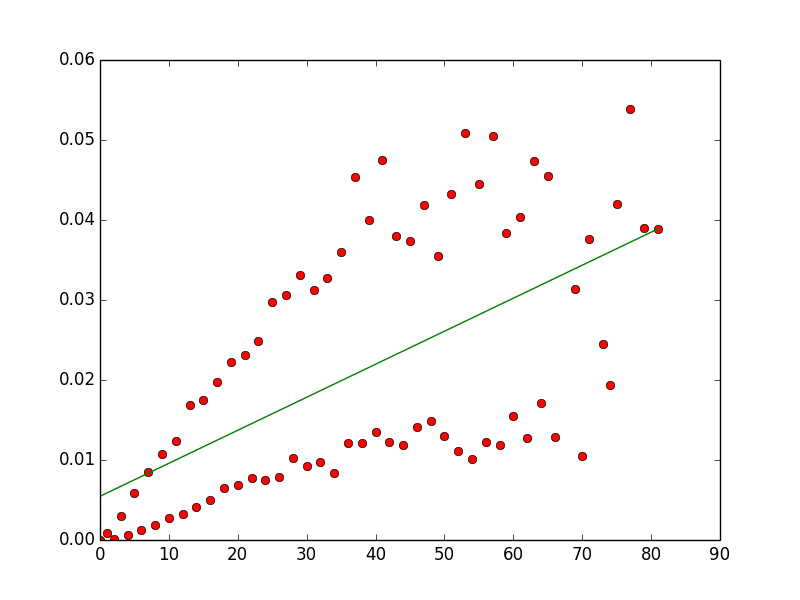}
            \caption{dataset: Facebook wall posts, Slope: 0.0004, intercept: .005, coefficient of determination: 0.36. Despite low coefficient of determination, we conclude linearity since we can observe two separate linear behaviors.}
        \end{subfigure}
        \caption{Plot of $f(d)/n(d)$ over $d$ as defined in Section \ref{sec:learnparams}. The green line shows the result of linear regression.} \label{fig:learn}
\end{figure}

\subsection{The densest subgraph algorithm in real dynamic graphs}
\label{subsec:4.3}
 In this section, our goal is to demonstrate the applicability of the rest and run strategy, in particular when it is  applied to the densest subgraph problem.  
  Thus, we compare the run time and accuracy of Algorithm \ref{algo:peel} and \cite{Epasto:2015} on the data sets presented in Table \ref{tab:realdata}. We note that all the networks of Table \ref{table:dataset} are not dynamic and not helpful for this section. We will later use them to generate synthetic data (Section \ref{Sec:expsyn}).
 Besides, these datasets only evolve through edge addition, thus we set $q = 0$ and we extract $p$ and $r$ through learning these parameters by running Algorithm \ref{algo:learn}.

\begin{table*}[!htb]
\center 
\small
\begin{tabular}{|c|c|c|c|c|c|c|c|c|c|c|}
\hline
data set & $n$ & $\vert E_{t_0}\vert$ & $d_{avg}$& $d_{max}$ & $\Delta_{t_0}$ & $\Gamma_0$ & $\beta_{t_0}$\\
\hline
Facebook (WOSN)& 63,731 & 481327 & 25.640 & 1,098 & 1758 & 33044752 &  91.268\\
\hline
Digg friends & 279,630  & 	1,731,653& 	12.385  & 	12,204 & 5276 & 97367436 & 57.34\\
\hline
YouTube & 3,223,589 & 	9,375,374 & 5.8167  & 	91,751 & 39021 & 1245457386 & 33.987\\
\hline
\end{tabular}
\caption{We used the above dynamic datasets from Koblenz dataset collection  in our experiments of Sections \ref{subsec:4.3} and \ref{Sec:expsyn}}
\label{tab:realdata} 
\end{table*}

\smallskip

\noindent\emph{Accuracy.}  Figure \ref{fig:dynamic} shows the result of running Algorithm \ref{algo:peel} on Digg friends dataset. We observe that algorithm tracks the density outputted by  \cite{Epasto:2015} closely, thus in practice the rest and run strategy does not  harm accuracy of the original algorithm.

\begin{figure}[!htb]
    \centering
\includegraphics[scale=0.36]{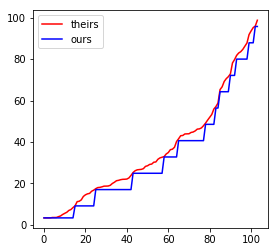}\hspace{1cm}
    \caption{Comparison of accuracy of our algorithm against \cite{Epasto:2015}  on Digg Friends dataset.  The $x$ axis is the number of rest and runs and the $y$ axis is the density output by each algorithm.}\label{fig:dynamic}
\end{figure}

 \noindent\emph{Run time.}
 In Table \ref{table:runtimereal}, we present the runtime of the algorithm for all datasets in Table \ref{tab:realdata}. The runtimes demonstrate that the rest and run strategy is highly effective in reducing when employed on huge graphs. In particular when the original algorithm does not terminate in reasonable time; see e.g. the row corresponding to YouTube.

\begin{table}[htb]
\center 
\small
\begin{tabular}{|c|c|c|}
\hline
data set & Our running time & \cite{Epasto:2015} \\
\hline
Facebook (WOSN)&4s& 12m\\
\hline
Digg Friends &20s& 1m\\
\hline
YouTube& 46s & +8h\\
\hline
\end{tabular}
\caption{Runtime comparison on real dynamic data sets.  The $x$ axis is the number of rest and runs and the $y$ axis is the density output by each algorithm.}
\label{table:runtimereal}
\end{table}

\subsection{The densest subgraph algorithm in synthetically generated data}\label{Sec:expsyn}

In this section, our goal is to study the efficiency and accuracy of  Algorithm \ref{algo:peel} in an ideal scenario when the graph's evolution is perfectly consistent with the assumptions of our model,i.e., we have perfect linearity. To achieve this, we take  $r,q=0.001, p=0.75$. We generate the edge sequence separately and  then we run dense graph detection algorithm; doing so we achieve two goals (1)  the time spent on the edge sequence generation has no affect in our algorithm runtime and (2)  our comparison is based on the same benchmark for our algorithm and the algorithms we compare against. 

In subsection \ref{subsec:4.1} we assess the quality of approximation by comparing the  density returned by our algorithm to the optimal density found by a linear program (LP) solver. We use an open source LP solver for our experiments. In Subsection \ref{subsec:4.2}, we switch to larger networks, having  vertex set of order ten thousand to a million, to compare the quality and run time of our algorithm to \cite{Epasto:2015}.

 \smallskip 
 
\textbf{Data Description.}
The datasets are collected from \cite{dataset2} and \cite{dataset1}, and are listed in Table \ref{table:dataset} and \ref{tab:realdata}.
We pre-process these datasets and remove duplicate edges and self loops, so for some data sets the number of edges in our input is slightly smaller than those presented in the original sources.
For dynamic graphs, we sort them according to their time-stamps in our pre-processing. 
A more detailed description of the data sets is presented in Appendix. 

\subsubsection{Small datasets: A comparison of accuracy  against optimal solution}\label{subsec:4.1} 
We evaluate the quality of our solution by comparing it to the the optimal answer. To obtain the maximum density, we solve the linear programming (LP) formulated by Charikar et al. \cite{Charikar:2000}. In Figure \ref{fig:ex1}, we present the two data sets for which we observed  the largest gap  between the outcomes of the two algorithms (ours and the LP). Note that the gap is always less than 1.15 after 40 rounds of rest-and-run.

\begin{figure}[htb]
\begin{center}
\includegraphics[scale=0.5]{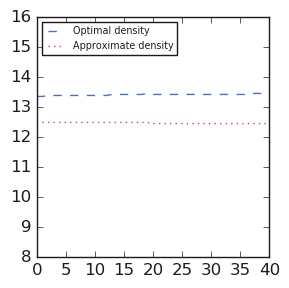}\hspace{3cm}
\includegraphics[scale=0.5]{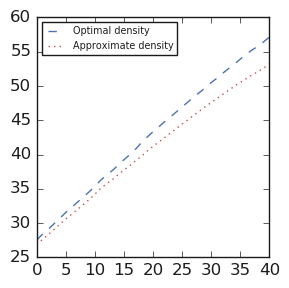}\\
\caption{Quality comparison against OPT \cite{Charikar:2000} on data sets: CA-CondMat (left),  email-Eu-core (right).  The $x$ axis is the number of rest and runs and the $y$ axis is the density output by each algorithm.}
\label{fig:ex1}
\end{center}
\end{figure}

\begin{table*}[htb]
\center 
\small
\begin{tabular}{|c|c|c|c|c|c|c|c|c|}
\hline
data set & $n$ & $\vert E_{t_0}\vert$ & $d_{avg}$& $d_{max}$ & $\Delta_{t_0}$ & $\Gamma_0$ & $\beta_{t_0}$\\
\hline
email-Eu-core & 986 & 16,064 & 32 & 345 & 626 & 1,183,216 & 26.86\\
\hline
CA-GrQc & 5241 & 14,484 & 5 & 81 & 147 & 229,867 & 19.56\\
\hline
CA-CondMat & 23,133 & 93,439 & 8 & 279 & 506 & 1,967,650 & 12.5\\
\hline

soc-Epinions & 75,879 & 405740 & 10 & 3044 & 5517 & 74,201,120 & 59.06\\
\hline 
twitter & 81,306 & 1,342,296 & 33 & 3383 & 6131 & 230,093,274 & 50.07\\
\hline
dblp & 317,080 & 1,049,866 & 6 & 343 & 622 & 21,780,889 & 48.60\\
\hline
Amazon  products & 334,863 & 925,872 & 5 & 549 & 995 & 9,752,186 & 2.76\\
\hline
soc-pokec-relationships & 1,632,803 & 22,301,964 & 27 & 14584 & 26919 & 2,086,073,558 &32.68\\
\hline
\end{tabular}
\caption{We use the above real world graphs which can be found in \cite{dataset2}.  }
\label{table:dataset}
\end{table*}

 \bigskip

\subsubsection{Big datasets: Runtime and quality assessment  on large data sets}\label{subsec:4.2}
We compare the efficiency of our algorithm to \cite{Epasto:2015}, which has amortized cost of a few  microseconds per update.

In order to make our comparison more accurate, we feed  the same sequence of updates  to   our algorithm  and the algorithm of \cite{Epasto:2015} for 100 rounds of rest and run. The only difference between our algorithm and algorithm of \cite{Epasto:2015} is that we apply  the changes after $\Delta_t$ updates at once to the graph and run $UpdatePeeling$. In contrast, \cite{Epasto:2015} calls $UpdatePeeling$ after any single update. We observe that the negligible  theoretical gap between the accuracy of our algorithm vs \cite{Epasto:2015} vanishes in practice, i.e., in our experiments  often both solutions have the same quality. In fact, as all of our plots demonstrate, the ratio of  our algorithm's output and the output of  \cite{Epasto:2015} is always at most $1.65$; see Figure \ref{fig:ex1}. 

Foremost, we observe a remarkable progress in the runtime of our algorithm. As we had expected, large values of $\Delta_t$s boost the total running time of our algorithm to the level that it outperforms algorithms of \cite{Epasto:2015} by factor of 1000.  The result of our comparison is presented in Table \ref{table:runtime}.

\begin{figure}[htb]
\begin{center}

\includegraphics[scale=0.5]{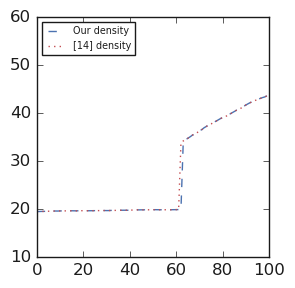}\hspace{3cm}
\includegraphics[scale=0.5]{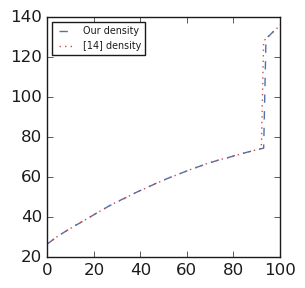}

\caption{Quality comparison against \cite{Epasto:2015} on CA-GrQc (left) and email-Eu-core (right) data sets. The $x$ axis is the number of rest and runs and the $y$ axis is the density output by each algorithm.}
\label{fig:ex2}
\end{center}
\end{figure}

\begin{table}[htb]
\center 
\small
\begin{tabular}{|c|c|c|}
\hline
data set & Our running time & \cite{Epasto:2015} \\

dblp & 3.43s & 32m \\
\hline
Amazon  products &  6s & 48m41s \\

\hline
soc-Epinions & 12.4s & 4h23m12s \\
\hline
twitter& 13.52s & 7h4m55s\\
\hline
soc-pokec-relationships & 2m 18s & +24h\\
\hline
\end{tabular}
\caption{Runtime comparison on larger data sets.}
\label{table:runtime}
\end{table}

\section{Acknowledgement}
Shahrzad Haddadan  was supported by  a Google Focused Research Award, SIR Grant RBSI14Q743, and NSF Award CCF-174074.  The authors are thankful to Flavio Chierichetti for several helpful conversations about this paper. 

\bibliographystyle{abbrv}

\bibliography{main}
\appendix
\subsection{Dataset description}\label{app:dataset-description}
Here we provide information on the data sets used in our experiments. This information is copied from pages on \cite{dataset1} describing the data set for completeness.

\medskip 

\textbf{Email-Eu-core network}
The network was generated using email data from a large European research institution. The data set has been anonymized. There is an edge $(u, v)$ in the network if person $u$ sent person $v$ at least one email. The e-mails only represent communication between institution members (the core), and the dataset does not contain incoming messages  from or outgoing messages to the rest of the world.

\medskip 

\textbf{Collaboration networks}
We have two networks from collaboration network datasets on \cite{dataset1}: CA-GrQc (General Relativity and Quantum Cosmology) and CA-CondMat(General Relativity and Quantum Cosmology). These networks are from e-print arXiv on the related field.  If an author $i$ co-authored a paper with author $j$, the graph contains a undirected edge from $i$ to $j$. If the paper is co-authored by k authors this generates a completely connected (sub)graph on $k$ nodes.

\medskip 

\textbf{Epinions social network}
This is a who-trust-whom online social network of a a general consumer review site Epinions.com. Members of the site can decide whether to ''trust'' each other. All the trust relationships interact and form the Web of Trust which is then combined with review ratings to determine which reviews are shown to the user.

\medskip 

\textbf{Twitter}
This dataset consists of 'circles' (or 'lists') from Twitter. Twitter data was crawled from public sources.

\medskip 

\textbf{DBLP collaboration network}
The DBLP computer science bibliography provides a comprehensive list of research papers in computer science. A co-authorship network is constructed by connecting two authors if they publish at least one paper together.
\medskip

\textbf{Amazon product co-purchasing network}
Network was collected by crawling Amazon website. It is based on "Customers Who Bought This Item Also Bought" feature of the Amazon website. If a product $i$ is frequently co-purchased with product $j$, the graph contains an undirected edge from $i$ to $j$.
\medskip

\textbf{Pokec social network}
Pokec is the most popular on-line social network in Slovakia. The popularity of network has not changed even after the coming of Facebook. Pokec has been provided for more than 10 years and connects more than 1.6 million people.

\subsection{Missing proofs}\label{app:missing-proofs}
In this section, we present the missing proofs from the main body.

\begin{proof}[Proof of Lemma \ref{lemma:tri}]

$$
\begin{array}{lll}
\mathbb{E}[(D_{t+\delta}(u)- D_{t}(u)) ]\leq\\ \sum_{i=1}^{\delta}\mathbb{E}[(D_{t+i}(u)- D_{t+i-1}(u)) ]
\leq \\
\sum_{i=1}^{\delta} \sum_{u\in V} (p_{u,v}(t+i-1))d_{t+i-1}(u,v)
\leq 
\\ 
\sum_{i=1}^{\delta} \sum_{u\in V} (d_{t+i-1}(u,v))^2p/K_t
 \\
\\
\text{~ Thus:}\\
\mathbb{E}[(D_{t+\delta}(u)- D_{t}(u)) ]\leq\\
(p/K_t) \sum_{i=1}^{\delta} \sum_{u\in V}   \left(d_t(u,v)+\frac{4(i-1)pd(u)d(v)}{K_t}\right)^2\leq \\
(p/K_t) \sum_{i=1}^{\delta} \sum_{u\in V}   \left(d_t(u,v)^2+\frac{16(i-1)^2p^2d^2(u)d^2(v)}{K^2_t}+\right.\\
\hspace{2 cm}
\left.\frac{8(i-1)pd(u)d(v)d_t(u,v)}{K_t}\right) 
\leq \\
(p/K_t) \sum_{i=1}^{\delta} \sum_{u\in V}   \left(d_t(u,v)^2+\frac{16(i-1)^2p^2d^2(u)d^2(v)}{K^2_t}+\right.\\
\hspace{2 cm}\left.\frac{8(i-1)pd(u)d(v)d_t(u,v)}{K_t}\right) 
\leq \\
(p/K_t) \sum_{i=1}^{\delta} \left( \zeta^2_{t}(v) +\frac{32(i-1)^2p^2d^2(v)K_t}{K^2_t}+  \frac{8(i-1)pm_td^2(v)}{K_t}\right) 
\leq \\
(p/K_t)  \zeta^2_{t}(v) \left( 1+\frac{32\delta^3p^2d^2(v)}{6K_t\zeta^2_{t}(v) }+  \frac{4 \delta^2 pm_td^2(v)}{K_t\zeta^2_{t}(v) }\right) 
\leq \\(p/K_t)\zeta^2_{t}(v)\delta(1+\frac{8 \delta p}{d_{avg}}+\frac{6\delta^2p^2}{K_t }).
\end{array}
$$

\end{proof}

\end{document}